\documentclass[conference]{IEEEtran}
\IEEEoverridecommandlockouts
\usepackage{cite}
\usepackage{amsmath,amssymb,amsfonts}
\usepackage{graphicx}
\usepackage{textcomp}
\usepackage{tabularray}
\usepackage{caption}
\usepackage{subcaption}
\usepackage{multicol,lipsum}
\usepackage{xcolor}
\usepackage{comment}
\def\BibTeX{{\rm B\kern-.05em{\sc i\kern-.025em b}\kern-.08em
    T\kern-.1667em\lower.7ex\hbox{E}\kern-.125emX}}

\definecolor{ref}{RGB}{0,50,160}

\usepackage{graphicx} 
\usepackage{amsmath}
\usepackage{amsfonts}
\usepackage{algorithm}
\usepackage{algorithmic}
\usepackage{amsmath}
\usepackage{amssymb}
\usepackage{amsthm}
\usepackage{booktabs}

\newtheorem{theorem}{Theorem}
\newtheorem{lemma}{Lemma}

\newtheorem{assumption}{Assumption}

\usepackage{verbatim}

\begin{document}
\title{Strategic Decision-Making Under Uncertainty through Bi-Level Game Theory and Distributionally Robust Optimization}
\author{Jiachen Shen, Jian Shi, Lei Fan, Chenye Wu, Dan Wang, Choong Seon Hong and Zhu Han

\thanks{J. Shen, J. Shi, L. Fan, and Z. Han are with the Departments of Electrical and Computer Engineering and Engineering Technology, University of Houston, Houston, TX, USA. C. Wu is with The Chinese University of Hong Kong, Shenzhen, China. D. Wang is with The Hong Kong Polytechnic University, Hong Kong. C. S. Hong is with Kyung Hee University, Seoul, South Korea. (e-mail: \{jshen24, jshi23@central, lfan8@central, zhan2\}@uh.edu, chenyewu@yeah.net, dan.wang@polyu.edu.hk, cshong@khu.ac.kr).
}}


\maketitle

\begin{abstract}
In strategic scenarios where decision-makers operate at different hierarchical levels, traditional optimization methods are often inadequate for handling uncertainties from incomplete information or unpredictable external factors. To fill this gap, we introduce a mathematical framework that integrates bi-level game theory with distributionally robust optimization (DRO), particularly suited for complex network systems. Our approach leverages the hierarchical structure of bi-level games to model leader-follower interactions while incorporating distributional robustness to guard against worst-case probability distributions. To ensure computational tractability, the Karush-Kuhn-Tucker (KKT) conditions are used to transform the bi-level challenge into a more manageable single-level model, and the infinite-dimensional DRO problem is reformulated into a finite equivalent. We propose a generalized algorithm to solve this integrated model. Simulation results validate our framework's efficacy, demonstrating that under high uncertainty, the proposed model achieves up to a 22\% cost reduction compared to traditional stochastic methods while maintaining a service level of over 90\%. This highlights its potential to significantly improve decision quality and robustness in networked systems such as transportation and communication networks.

\end{abstract}

\section{Introduction}
In the realm of optimization theory, recent advancements have significantly transformed decision-making processes, particularly under uncertainties. The increasing complexity of modern systems, driven by factors such as globalization, technological advancements, and the dynamic nature of markets, has necessitated the development of more sophisticated optimization techniques. A notable development in this field is the rise of DRO, which has gained prominence due to its effectiveness in handling uncertainties across various domains, such as finance, supply chain management, and energy systems \cite{bard1983algorithm, ben-tal2015robust}. DRO’s strength lies in its ability to account for the worst-case scenarios over a set of probability distributions, making it a powerful tool for robust decision-making under uncertain conditions \cite{sinha2017evolutionary}. This approach not only mitigates the risks associated with adverse outcomes but also provides decision-makers with confidence in the stability and resilience of their solutions, even in the face of highly variable or incomplete data. 

Simultaneously, game-theoretic frameworks have become critical for modeling strategic interactions. Within this domain, bi-level optimization has established itself as a key tool for hierarchical decision-making. This approach models scenarios where a leader's decisions influence a follower's responses, creating a nested optimization structure that is particularly useful in power systems, communication networks, and transportation networks. As a general mathematical framework, it is important to distinguish it from the classical Stackelberg game, a well-known economic instance. The broader term `bi-level optimization' is preferred when the focus extends beyond specific economic models to general mathematical methods---such as integration with DRO---and to applications in diverse domains like the network systems previously mentioned.

Beyond this leader-follower structure, other advanced game-theoretic strategies are also being developed to enforce robust cooperation under uncertainty. In particular, zero-determinant strategies---originally formulated for repeated games---have been extended to finite games with implementation errors, enabling a player to unilaterally enforce linear payoff relationships even when actions are executed imperfectly \cite{ZD_Error2021}. Such strategies have been applied in group decision-making to sustain cooperation in noisy environments \cite{ZD_Cooperation2022}. Concurrently, studies on the co-evolutionary dynamics of threshold public goods games with collective-risk environment feedback reveal how populations adapt when contributions and risk thresholds interact under stochastic shocks \cite{PGG_CoEvo2023}. These developments highlight the relevance of robust strategic design in uncertain multi-agent systems, motivating our bi-level distributionally robust optimization approach.

In real world applications, decision makers face interdependent decision and uncertainties. The integration of bi-level optimization and DRO provides a new way to make the optimal decision in these practical applications. However, the DRO bi-level optimization framework is under explored. The primary challenge in this integrated framework inherent the complexity of bi-level optimization, which is computationally demanding because of issues such as non-convexity and difficulty in finding the global optimal solution. Incorporating DRO compounds these challenges by necessitating the consideration of worst-case distributions, further complicating the solution landscape \cite{yang1995heuristic}. Despite these challenges, the potential benefits of this integration are significant, offering a robust approach that can handle both the hierarchical structure of decision-making and the deep uncertainties inherent in real-world problems. The integrated framework not only enhances the robustness of the decision-making process but also provides a more comprehensive approach to managing risk in complex systems.

In this paper, our proposed framework is formulated as a bi-level optimization problem, which serves as the mathematical representation of a Stackelberg game, which extends the classical bi-level framework to handle profound uncertainty. By integrating DRO, our model manages ambiguity in the underlying probability distributions, deriving strategies that are resilient to shifts in the stochastic process itself. This integrated approach equips decision-makers to more effectively navigate complex hierarchical systems and ensure their strategies are adaptable to potential disruptions.

Developing a tractable bi-level DRO framework requires overcoming significant technical challenges. The primary hurdle is integrating the intrinsically difficult, nested structure of bi-level optimization with the infinite-dimensional nature of DRO. This combination creates a complex, large-scale optimization problem for which standard solvers are inadequate. To address these hurdles, this paper makes several important contributions:

\begin{itemize}
    \item We propose a novel framework that models the strategic interactions in networked systems as a bi-level game, while ensuring the leader's decision is robust against worst-case distributional shifts. To the best of our knowledge, this is the first work to address this specific intersection of hierarchical games and distributional robustness in such systems.

    \item We introduce a unique MPEC (Mathematical Program with Equilibrium Constraints) formulation where the equilibrium constraints themselves contain a nested, dualized Wasserstein DRO problem. We develop a tailored algorithmic solution combining proximal and cutting-plane methods to solve this non-smooth, bi-level formulation.

    \item We provide theoretical insights into the convergence properties and computational efficiency of the proposed model. By employing dual reformulation techniques, we ensure the computational feasibility of the framework, with proofs of convergence under certain conditions.

    \item We demonstrate the practical utility of the proposed framework in various network systems, such as transportation and communication networks. The model’s robustness and adaptability are validated through numerical experiments, which highlight its effectiveness in managing uncertainty.
\end{itemize}

The remainder of this paper is organized as follows: Section \ref{sec2} reviews the related literature on bi-level optimization and distributionally robust optimization, highlighting key advancements and gaps. In Section \ref{sec3}, we present the model formulation, detailing the integration of bi-level optimization with DRO, including the mathematical underpinnings and key assumptions. Section \ref{sec4} outlines the algorithmic framework developed to solve the integrated model, with a focus on the computational techniques employed. In Section \ref{sec5}, we conduct simulation experiments to validate the proposed framework, demonstrating its practical applicability and robustness. Finally, Section \ref{sec6} concludes the paper with a summary of findings, implications for future research, and potential real-world applications.
For clarity and easy reference, the key notation used throughout this paper is summarized in Table \ref{tab:notation}.
\begin{table}[h!]
\caption{Summary of Notation}
\label{tab:notation}
\centering
\begin{tabular}{@{}cl@{}}
\toprule
\textbf{Symbol} & \textbf{Description} \\
\midrule
\multicolumn{2}{l}{\textit{Decision Variables}} \\
$x$ & Leader's decision variables \\
$y$ & Follower's decision variables \\
\midrule
\multicolumn{2}{l}{\textit{Problem Functions}} \\
$f(x,y,\xi)$ & Leader's objective function \\
$g_i(x,y,\xi)$ & Leader's constraint functions \\
$h(x,y,\xi)$ & Follower's objective functional \\
$k_j(x,y,\xi)$ & Follower's constraint functions \\
\midrule
\multicolumn{2}{l}{\textit{Uncertainty and DRO}} \\
$\xi$ & Vector of uncertain parameters \\
$\Omega$ & Support set of the uncertainty $\xi$ \\
$P$ & A probability distribution of $\xi$ \\
$\mathcal{P}$ & Ambiguity set of probability distributions \\
$\hat{\mathbb{P}}_N$ & Empirical distribution from N samples \\
$\epsilon$ & Radius of the Wasserstein ambiguity set \\
$d_W$ & The Wasserstein distance metric \\
\midrule
\multicolumn{2}{l}{\textit{Optimization and Algorithm}} \\
$\mu_j$ & Lagrange multipliers for follower's constraints $k_j$ \\
$t, T$ & Iteration counter and maximum iterations \\
$\eta_t$ & Proximal parameter at iteration $t$ \\
$L(\cdot)$ & The Lagrangian function \\
\bottomrule
\end{tabular}
\end{table}

\section{Literature Review}\label{sec2}
\subsection{Foundational Concepts in Hierarchical and Robust Optimization}
Modern decision-making in complex systems often involves navigating both structural hierarchies and parametric uncertainty. Bi-level optimization provides a powerful framework for modeling hierarchical scenarios where a leader's decisions influence a follower's optimal response, creating a nested problem structure \cite{sun2022bi-level,liao2022privacy}. This approach is particularly applicable to network systems, such as supply chains and smart grids \cite{mhaisen2022optimal,wang2020distributionally}. A standard method for solving such problems involves using the KKT conditions to transform the follower's problem, effectively converting the bi-level model into a more tractable single-level problem \cite{sun2022bi-level,liao2022privacy}. Complementing this structural approach, DRO offers a sophisticated paradigm for handling uncertainty by optimizing against a worst-case scenario over an entire set of possible probability distributions \cite{9966489,9107344}. The appeal of DRO lies in its ability to produce solutions that are resilient to model misspecification, leading to its wide application in domains like finance and energy systems \cite{10478176, delage2010distributionally}. The integration of these two paradigms---bi-level optimization for structural hierarchy and DRO for profound uncertainty---creates the robust framework for strategic decision-making that this paper explores.

\subsection{Advances in Integrating Bi-Level Optimization with Distributionally Robust Optimization}
While bi-level optimization and DRO are individually well-established, their direct integration is an active and challenging research frontier, motivated by the need to solve hierarchical decision-making problems under deep uncertainty \cite{bental2004robust}. Early explorations into this domain often relied on significant simplifying assumptions, such as restricting the uncertainty to specific, simple parametric forms or assuming that uncertainty only affects the follower's problem. These assumptions, while enabling initial formulations, limited the applicability of the models to more complex, real-world scenarios where data is incomplete and distributional forms are unknown.

More recent research has begun to tackle these challenges by incorporating more sophisticated uncertainty models. For instance, some works have explored bi-level formulations in the context of decentralized energy management, considering factors like grid reliability and agent misbehavior under uncertainty \cite{bertsimas2011robust}. Concurrently, advancements in data-driven DRO, particularly methods using the Wasserstein metric to define ambiguity sets, have provided powerful new tools for single-level optimization \cite{luo2019decomposition, EsfahaniKuhn2018}. While general solution methods for convex bi-level problems continue to advance \cite{Nguyen2023}, embedding these advanced DRO formulations within such a structure is highly non-trivial. The primary difficulty lies in the analytical and computational complexity of solving an optimization problem where the constraints themselves contain a nested, non-smooth, worst-case expectation problem. As a result, the development of tractable frameworks for bi-level games under modern, data-driven DRO, especially for problems with coupled constraints like those found in integrated energy systems \cite{yan2023bilevel}, remains an area of active investigation. This paper contributes to this frontier by developing a solvable method specifically for this class of challenging problems.

\begin{figure}[t]
    \centering
    \includegraphics[width=\linewidth]{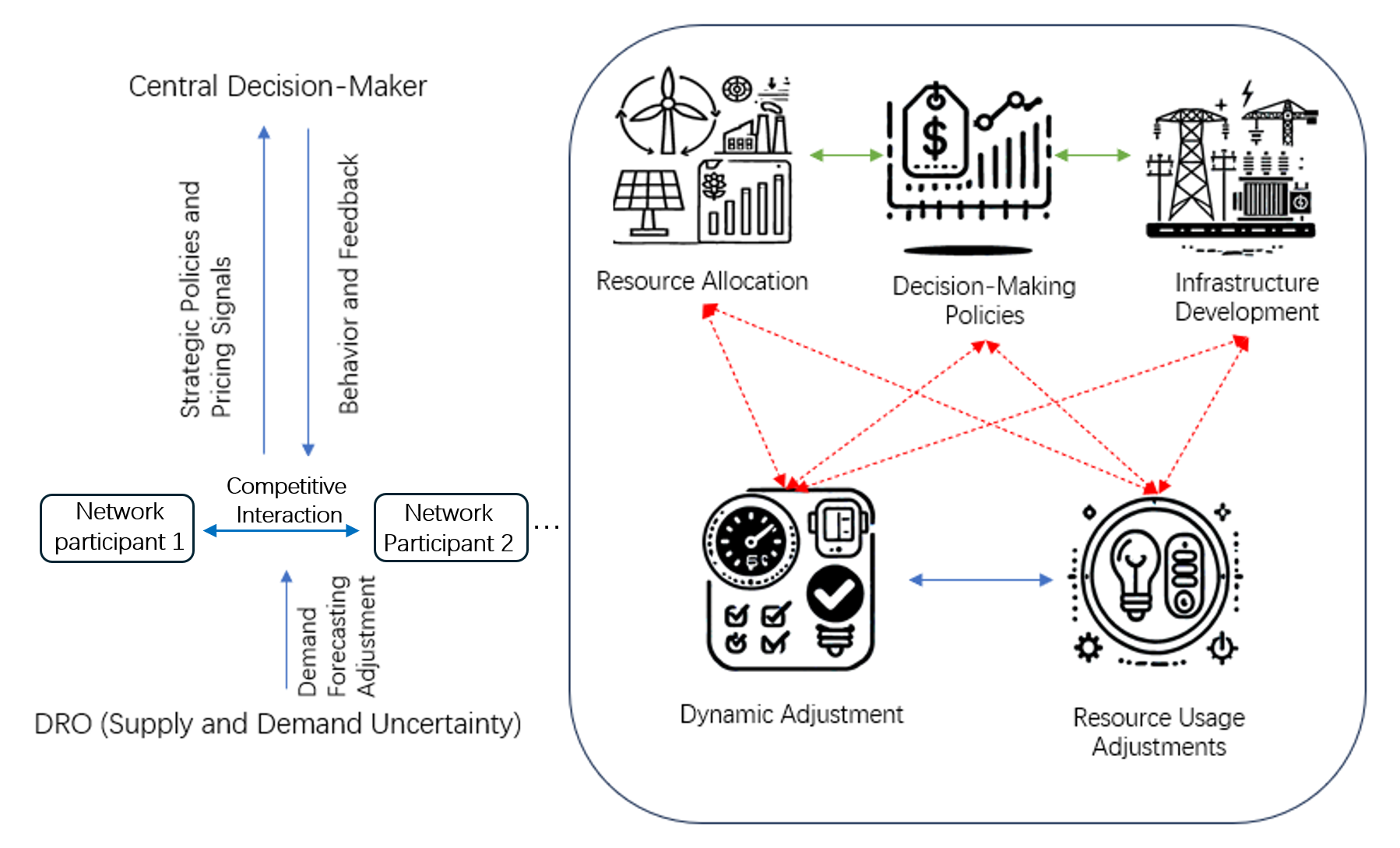}
    \caption{Applications of DRO in General Networks.}
    \label{fig:DRO_applications}
\end{figure}

\subsection{Applications in Complex Network Systems}
The hierarchical and uncertain nature of complex network systems makes them a prime application area for the advanced optimization frameworks discussed above. The leader-follower structure inherent in bi-level optimization is a natural fit for modeling interactions in telecommunications, supply chains, and smart grids, where central operators or planners set policies that influence the behavior of individual users or local agents \cite{mhaisen2022optimal, wang2020distributionally}. For instance, strategic leader-level decisions, such as robust operational planning or optimal sensor placement for monitoring, are critical for ensuring overall system resilience and efficiency \cite{zi2022distributionally, zhang2020robust}. In many of these systems, pricing mechanisms are a key tool used by the leader to manage network resources and guide follower behavior \cite{raveendran2022pricing}.

Simultaneously, distributionally robust optimization has been increasingly applied to network problems to manage inherent uncertainties. Specific applications include designing robust supply chain networks under demand uncertainty \cite{wang2020distributionally} and optimizing maintenance schedules in power systems against failures \cite{zheng2023distancebased}. These works demonstrate the value of DRO in creating reliable network solutions. However, many real-world network problems exhibit both hierarchical decision structures and deep uncertainty, creating a need for integrated models. For example, recent applications include employing bi-level optimization for reactive power management in integrated electricity-gas systems \cite{Zhao2021}. A network operator must set robust policies or prices in the face of uncertain market conditions, while multiple users or sub-systems react competitively, a dynamic explored in game-theoretic resource allocation models \cite{zhang2020distributed, raveendran2019vlc}.

The complex interplay between a central decision-maker and network participants, as discussed in the literature, is conceptually illustrated in Figure 1. This figure encapsulates the dynamics where a leader's high-level strategic decisions on key operational areas---such as Resource Allocation, the formulation of Decision-Making Policies, and long-term Infrastructure Development---are transmitted downwards to followers via "Strategic Policies and Pricing Signals." In response, the followers engage in competitive interaction, leading to Dynamic Adjustment and Resource Usage Adjustments. The outcomes of these actions are then conveyed back to the leader as "Behavior and Feedback." Our work is situated within this context, addressing the critical gap where both the hierarchical (bi-level) and DRO aspects of this interaction must be handled simultaneously to achieve truly resilient and efficient network management.
\section{Model Formulation}\label{sec3}
In this section, we present the mathematical framework. Our model is conceptualized as a two-stage decision problem under uncertainty. In our game-theoretic context, the "payoff" for each player is defined by their respective objective. The leader aims to minimize a system-wide cost, so their payoff is the negative of their objective function, while the follower minimizes local costs. The effectiveness of this payoff design in handling uncertainty stems directly from our DRO approach. Rather than optimizing against a single assumed probability distribution, our framework seeks to find a solution that performs best under the worst-case distribution from an entire set of possibilities. This formulation is essential for capturing the dynamic interactions between leaders and followers while ensuring decisions are robust against deep uncertainty. We begin by detailing the core mathematical components of the model and explicitly formulating the leader-follower dynamics. To illustrate the practical relevance of our theoretical development, we then provide a detailed simulation scenario description.

\subsection{Upper-Level Problem (Leader)}
The leader's optimization problem is expressed as:
\[
\min_{x\in\mathbb{R}^{n}} \max_{P\in\mathcal{P}} \mathbb{E}_{P}[f(x,y,\xi)]
\]
\[
s.t.\ \ g_{i}(x,y^{*}(x),\xi)\le0, \quad \forall \xi \in \Xi, \quad \forall i
\]

\noindent where the constraint must hold robustly for all realizations of the uncertainty \(\xi\) within its support set \(\Xi\),  \(x \in \mathbb{R}^n\) represents the leader's decision variables, \(y^*(x)\) is the follower's optimal response function, \(\xi\) is the vector of uncertain parameters, \(g_i\) are the leader's constraint functions, and \(\mathbb{E}_P[\cdot]\) denotes the expectation taken with respect to a probability distribution \(P\) from the ambiguity set \(\mathcal{P}\). The set \(\mathcal{P}\) represents the ambiguity set of probability distributions. The function \(f\) is presumed to be convex with respect to the uncertain parameter \(\xi\).

\subsection{Lower-Level Problem (Follower)}
Parameterized by the leader's decision \( x \), the follower's optimization problem minimizes a cost functional \( h \) subject to constraints \( k_j \):

\begin{equation}
\begin{aligned}
& \max_{y \in \mathbb{R}^m} h(x, y, \xi) \\
& \text{s.t.} \quad k_j(x,y,\xi)\leq 0,\ \forall j,
\end{aligned}
\end{equation}

\noindent Here, \( y \) represents the follower's decision variables and \( y \in \mathbb{R}^m \). The constraints \( k_j \) embody the restrictions inherent to the follower's problem.

\subsection{Integrative Formulation}

To integrate the follower's problem into the leader's optimization framework, we utilize the KKT conditions as constraints for the reformulated single-level problem with DRO:
\begin{align}
    &\nabla_yh(x,y,\xi)+\sum_i\mu_i\nabla_yk_i(x,y,\xi)=0\label{Stationary}\\
    &k_j(x,y,\xi)\leq0,\ \forall j\label{PrimalFeasible2}\\
    &\mu_j\geq0,\ \forall j\label{Dual Feasible}\\
    &0\leq\mu_i k_j(x,y,\xi)\geq0,\ \forall j,\label{Complementary Slackness}
\end{align}
\noindent where (\ref{Stationary}) denotes the stationary condition, (\ref{PrimalFeasible2}) and (\ref{Dual Feasible}) denote the feasibility of primal and dual problem, (\ref{Complementary Slackness}) denotes the complementary slackness condition. The single-level reformulation encapsulates the leader's DRO objective and the follower's optimality conditions within a single optimization framework:
\begin{equation}
\label{eq:1}
\begin{aligned}
    \min_{x\in X, y \in Y}&\max_{\mathbb{P}\in\mathcal{B}_\varepsilon (\hat{\mathbb{P}}_N)} \mathbb{E}_{\mathbb{P}}[f(x,y,\xi)] \\
    \text{s.t.} \quad &(2)-(5).
\end{aligned}
\end{equation}
The ambiguity set $\mathcal{B}_\varepsilon (\hat{\mathbb{P}}_N)$ is  defined on the space $\mathcal{M}(\mathcal{E})$ of all probability distributions $\mathbb{P}$ supported on $\Omega$ with
$\mathbb{E}_{\mathbb{P}}\left[\|\xi\|\right] = \int_{\Omega} \|\xi\| \mathbb{P}(d\xi) < \infty$ with the formula

\[
\mathcal{B}_\varepsilon (\hat{\mathbb{P}}_N) := \left\{ \mathbb{P} \in \mathcal{M}(\Omega) : d_W(\hat{\mathbb{P}}_N, \mathbb{P}) \leq \varepsilon \right\},
\]
\noindent where $d_W$ represents the Wasserstein distance. This ambiguity set  can be viewed as the Wasserstein ball of radius $\varepsilon$ centered at the empirical distribution $\hat{\mathbb{P}}_N$.
\subsection{Dual Reformulation}
The worst-case expectation problem \eqref{eq:1} constitutes an infinite-dimensional optimization problem over probability distributions, making it seemingly intractable (In this
context, a problem is considered “tractable” if it can be solved to optimality in a time that is polynomial in the input size, typically by reformulating it as a finite-dimensional
convex program for which efficient algorithms exist). However, leveraging the dual reformulation techniques commonly used in robust optimization \cite{zhang2017multi}, this problem can be re-expressed as a finite-dimensional convex program. Key theoretical results, such as convex reduction, provide the necessary conditions under which this dual reformulation is applicable, ensuring the solution's optimality and feasibility \cite{zhao2021reactive}. The convexity assumption plays a crucial role in guaranteeing that the transformed problem remains within the realm of convex optimization, thus making the application of standard solvers feasible \cite{aghamohammadi2024decentralized}.

The tractability results in the remainder of this paper are predicated on the following convexity assumption.

\begin{assumption}[Convexity] The uncertainty set $\Omega\subseteq \mathbb{R}^m$ is convex and closed. The objective function \(f(x,y,\xi)\), the follower's cost functional \(h(x,y,\xi)\), and the constraints \(g_i(x,y,\xi)\) and \(k_j(x,y,\xi)\) are convex with respect to the \textbf{decision variables \((x,y)\)} for any fixed value of \(\xi\).
\end{assumption}
While Assumption 1 is critical for ensuring tractability, we acknowledge it may not hold universally, as real-world problems can feature non-convexities from integer variables (e.g., facility location choices). However, adopting this assumption is a deliberate and standard choice that allows us to establish a foundational framework and analyze the core complexities of the bi-level DRO structure. This framework remains directly applicable to a significant class of problems where convexity is a reasonable premise, including many resource allocation, pricing, and operational planning scenarios.

Then, by leveraging dual variables $\lambda$ for the follower's problem, we reduce the above problem into a single-level convex optimization model \cite{tsai2024distributionally} as the following theorem:

\begin{theorem}[Convex reduction]
If the convexity Assumption 1 holds, then for any \(\varepsilon \geq 0\) the worst-case expectation (8) equals the optimal value of the finite convex program
\begin{equation*}
\begin{aligned}
& \underset{x, y, \lambda, s_i, z_{i}, v_{i}}{\text{inf}} \lambda \varepsilon + \frac{1}{N} \sum_{i=1}^{N} [-f]^*(x,y,z_{i} - v_{i}) + \sigma_{\Omega}(v_{i}) - \langle z_{i}, \hat{\xi}_i\rangle  \\
& \text{s.t.}\quad\|z_{i}\|_* \leq \lambda ,\quad \forall i.\\
&\quad\quad(2)-(5),
\end{aligned}
\end{equation*}
\end{theorem}
\noindent where \(s_i\) is the auxiliary slack variable introduced to handle the reformulation and assists with the linearization of the problem, \(z_i\) are dual variables associated with the constraints, capturing the sensitivity of the objective function to changes in uncertainty \( \xi \), \(v_i\) is the  variable related to the uncertainty set \( \Omega \), ensuring that the uncertainty over probability distributions is properly handled. The term \([-f]^*(x,y,z_{i} - v_{i})\) denotes the conjugate of \(-f\) with respect to \(\xi\), and \(\|z_{i}\|_*\) represents the dual norm of \(z_{i}\). Moreover, \(\chi_\Omega\) represents the characteristic function of \(\Omega\), and \(\sigma_\Omega\) is its conjugate, that is, the support function of \(\Omega\). follows a similar approach to the one presented in section 4.2 in \cite{liang2024learning}.

The dual reformulation enables the use of conventional optimization algorithms by transforming the inherently complex bi-level structure into a tractable format. To demonstrate the practical application of the techniques outlined in the model formulation, we present a simulation scenario that showcases the model's effectiveness in a realistic context.

\subsection{Simulation Scenario Description}
Our formulation directly embeds the structural properties of the target networked system. For example, in the supply chain application detailed below, the follower’s constraints explicitly model the flow conservation laws at each node, and the transportation costs in the leader’s objective function are dependent on the network topology. This embedding ensures that our solutions are not only statistically robust but also physically feasible and directly applicable to real-world infrastructure systems.

In this scenario, a central planner and facility managers work together to minimize logistics costs amid uncertain demand. While typical supply chain networks involve integer or binary decision variables (e.g., facility openings or routing decisions), which would introduce non-convexities, we intentionally focus on continuous decision variables such as inventory levels and transportation quantities. This modeling choice is a necessary simplification designed to isolate the core scientific challenge of this paper: understanding the fundamental interplay between hierarchical game dynamics and distributional uncertainty. By maintaining a convex and tractable problem structure, we can clearly demonstrate the unique value and behavior of our bi-level DRO framework without the confounding effects of combinatorial complexity. Incorporating integer variables is a valuable, yet distinct, direction for future work.
\subsubsection{Upper-Level Problem}
The objective of the central planner is to minimize the total expected logistics cost. This cost includes inventory holding costs, transportation costs, and penalties for unmet demand. The problem can be mathematically formulated as:
\[
\begin{aligned}
&f(x, y) = \sum_{i \in I} c_i x_i + \sum_{(i,j) \in A} t_{ij} y_{ij} 
\\&+ \sup_{P \in \mathcal{P}} \mathbb{E}_{\xi \sim P} \{\sum_{i \in I} p_i \max (0, d_i(\xi) - x_i - \sum_{j:(i,j) \in A} y_{ij} ) 
\\&+ g(z^*(x, y, \xi))\},
\end{aligned}
\]
\noindent where \( x_i \) represents the inventory level at location \( i \), \( y_{ij} \) denotes the quantity transported from location \( i \) to location \( j \), \( c_i \) is the cost per unit of inventory, \( t_{ij} \) is the transportation cost per unit, \( p_i \) is the penalty cost for unmet demand, and \( d_i(\xi) \) represents the random demand at location \( i \). The term \( g(z^*(x, y, \xi)) \) refers to the cost associated with the optimal response of the followers.

The constraints for this problem are:
\[
0 \leq x_i \leq C_i, \quad \forall i \in I
\]
\[
0 \leq y_{ij} \leq T_{ij}, \quad \forall (i,j) \in A,
\]
\noindent where \( C_i \) is the storage capacity at location \( i \), and \( T_{ij} \) is the transportation capacity between locations \( i \) and \( j \).

\subsubsection{Lower-Level Problem}
The objective of the facility managers is to minimize local operational costs by optimally allocating resources based on the central planner's decisions. The problem is formulated as:
\[
g(x, y, z) = \sum_{i \in I} q_i z_i + \sum_{(i,j) \in J} r_{ij} z_{ij},
\]
\noindent where \( z_i \) represents the resource allocation at location \( i \), \( q_i \) is the cost associated with resource allocation at \( i \), and \( r_{ij} \) is the cost associated with resource flow between locations \( i \) and \( j \).

The constraints for the follower's problem are:
\[
0 \leq z_i \leq R_i(x_i, y_{ij}, \xi), \quad \forall i \in I
\]
\[
\sum_{j:(i,j)\in A} z_{ij} = \sum_{j:(j,i)\in A} z_{ji} + d_i(\xi) - x_i, \quad \forall i \in I,
\]
\noindent where \( R_i(x_i, y_{ij}, \xi) \) is the maximum feasible resource allocation at location \( i \), depending on the inventory levels \( x_i \), transportation quantities \( y_{ij} \), and realized demand \( \xi \).

\subsubsection{Integration of the Bi-Level Problem with DRO and Concrete Formulation}

To solve the problem more efficiently, the bi-level problem is reformulated into a single-level problem.

The reformulated objective function is:
\begin{equation}
\begin{aligned}
\min_{x,y,z,\lambda} \sup_{P \in \mathcal{P}} \mathbb{E}_{\xi \sim P} [ \sum_{i \in I} \left(c_i x_i + q_i z_i \right) + \sum_{(i,j) \in A} \left(t_{ij} y_{ij} + r_{ij} z_{ij}\right) 
\\+ \sum_{i \in I} p_i \max (0, d_i(\xi) - x_i - \sum_{j:(i,j)\in A} y_{ij} ) ],
\end{aligned}
\end{equation}

The constraints for the single-level problem include:
\[
\frac{\partial h(x, y, z, \xi)}{\partial z_i} + \sum_{j} \mu_j \frac{\partial k_j(x, y, z, \xi)}{\partial z_i} = 0, \quad \forall i \in I,
\]
\[
k_j(x, y, z, \xi) \leq 0, \quad \forall j,
\]
\[
\mu_j \geq 0, \quad \forall j,
\]
\[
\mu_j k_j(x, y, z, \xi) = 0, \quad \forall j,
\]
In this reformulated problem, \( \lambda_j \) represents the Lagrange multipliers associated with the constraints \( g_j(x, y, z, \xi) \), ensuring that the follower’s decisions are optimal and feasible given the leader’s decisions.

\section{Algorithmic Framework}\label{sec4}
We develop an algorithmic framework to address the challenges posed by the bi-level DRO model depicted in Equation (8). Our approach integrates primal-dual updates with cutting-plane methods to handle the bi-level structure and distributional uncertainty effectively.

For the development of our algorithm, we impose the following assumptions to ensure mathematical tractability and convergence:
\begin{itemize}
    \item Continuity and Differentiability: All functions are assumed to be continuous and possess continuous first derivatives with respect to $(x, y)$.
    \item Compactness: The decision spaces $X$ and $Y$, as well as the space of dual variables $\lambda$, are assumed to be non-empty, closed, and bounded.
    \item Strong Duality: For the lower-level problem, we assume that strong duality holds, allowing us to work with the dual problem without a duality gap.
\end{itemize}

The following subsections provide a detailed explanation of the methods and techniques employed in the algorithmic framework. We begin by discussing the proximal dual update, which enhances the stability of dual ascent algorithms. Next, we introduce the cutting-plane primal update method, which is essential for handling non-differentiable terms within the objective function. We then outline the complete algorithm steps, detailing the sequence of operations necessary to solve the bi-level DRO problem. Finally, we conclude with a Convergence Analysis that offers theoretical guarantees for the stability and optimality of the proposed solution.

\subsection{Proximal Dual Update}

The proximal update method serves to enhance the stability of dual ascent algorithms in scenarios with non-smooth or poorly conditioned problems. For the bi-level DRO model, a proximal term augments the dual objective function, adding a regularization effect that promotes robust convergence behavior.

Incorporating the proximal regularization into the dual update, we adjust the dual variable vector $\lambda$ at iteration $t$ by solving the following proximal problem:
\begin{equation}
\lambda^{t+1} = \arg \min_{\lambda \geq 0} \left( L(x^t, y^t, z^t, v^t, \lambda) + \frac{1}{2\eta_t}\|\lambda - \lambda^t\|^2 \right)
\label{eq:2}
\end{equation}
\noindent where:
\begin{itemize}
    \item $L(x, y, z, v, \lambda)$ is the Lagrangian associated with the primal problem, expressed as:
    \begin{equation}
    L(x, y, z, v, \lambda) = f(x, y) + \lambda \left(g(x, y) - z\right) + v \cdot h(x, y),
    \label{eq9}
    \end{equation}
    where $f(x, y)$ is the objective function, $g(x, y)$ represents the constraints, and $h(x, y)$ accounts for the dual variables associated with constraint violations.
    \item $\eta_t$ represents the proximal parameter at iteration $t$.
    \item $\lambda^t$ is the current dual variable vector.
\end{itemize}

The term $\frac{1}{2\eta_t}\|\lambda - \lambda^t\|^2$ ensures the updated dual variables do not deviate excessively from the current iteration, thereby preventing oscillatory behavior and facilitating the convergence toward a dual optimal solution.

To provide a more intuitive understanding of this update, we can interpret \ref{eq9} from both an abstract and a practical standpoint. Conceptually, the equation represents the leader’s robust decision rule: before committing to a strategy \(x\), the leader anticipates the follower’s best response \(y\) and the most adverse distributional perturbation within the ambiguity set. In practical terms—e.g., in a supply chain—this is akin to planning for the highest possible penalty cost that could occur if demand patterns shift unfavorably, ensuring the chosen plan remains viable under worst-case deviations from historical data.
\subsection{Cutting-Plane Primal Update}
To manage the non-differentiable supremum term \( \sup_{P \in \mathcal{P}} \mathbb{E}_P[f(x, y, \xi)] \) within the objective function of the bi-level DRO model, we utilize the cutting-plane method. This method iteratively refines a piecewise-linear approximation of the non-differentiable term.

The primal update at iteration $t$ involves solving a linear program that approximates the supremum term by constructing a series of linear pieces:
\begin{equation}
\label{eq11}
\begin{aligned}
& \min_{x \in X, y \in Y} \left( \lambda^t \epsilon + \frac{1}{N} \sum_{i=1}^{N} s_i \right) \\
& \text{s.t.} \\
& s_i \geq [-f]^*(x,y,z_{i} - v_{i}) + \sigma_{\Omega}(v_{i}) - \langle z_{i}, \hat{\xi}_i\rangle, \quad \forall i, \\
&\|z_{i}\|_* \leq \lambda ,\quad \forall i,\\
& \nabla_y h(x, y, \hat{\xi}) + \sum_j \mu_j^t \nabla_y k_j(x, y, \hat{\xi}) = 0, \quad \forall \hat{\xi}, \\
& \mu_{j}^{t} k_{j}(x,y,\hat{\xi}) = 0,\quad \forall j, \forall\hat{\xi} \\
& g_i(x, y, \hat{\xi}) \leq 0, \quad k_j(x, y, \hat{\xi}) \leq 0, \quad \forall j, \hat{\xi}.
\end{aligned}
\end{equation}
Here, $s_i$ represents the linear approximation of the supremum term for scenario $\hat{\xi}_i$. The constraint $s_i \geq [-f]^*(x,y,z_{i} - v_{i}) + \sigma_{\Omega}(v_{i}) - \langle z_{i}, \hat{\xi}_i\rangle$ ensures that the approximation is a lower bound on the actual supremum, maintaining feasibility and optimality of the solution.

 At each iteration, the dual variable $\lambda$ is updated by solving the problem represented by \eqref{eq:2}. The Lagrangian \( L(x^t, y^t, z^t, v^t, \lambda) \) includes the objective function $f(x, y)$, the constraints $g(x, y)$, and dual variables $\lambda$ and $v$. The proximal regularization term $\frac{1}{2\eta_t}\|\lambda - \lambda^t\|^2$ ensures that the updated dual variables do not deviate excessively from the current iteration, promoting stability and convergence. The cutting-plane primal update refines the approximation of the supremum term while the proximal dual update adjusts the dual variables, ensuring the stability and convergence of the overall bi-level DRO model. Both updates work in tandem, with the dual variables influencing the primal update through the objective function and constraints, and vice versa.

\subsection{Algorithm Steps}
The following steps outline the algorithm to solve the given bi-level DRO model as shown in Algorithm~\ref{alg1}:

\begin{enumerate}
    \item Initialization: Set $x^0 \in X$, $y^0 \in Y$, $\lambda^0 \geq 0$, and iteration counter $t = 0$.
    \item Primal Update: For the current $\lambda^t$, solve the primal problem while linearizing the supremum term to obtain the next iterates $x^{t+1}$ and $y^{t+1}$.
    \item Supremum Evaluation: Calculate the supremum term using the updated primal variables for the objective function approximation.
    \item Dual Update: Apply the proximal point method to update $\lambda^{t+1}$ by solving the corresponding dual problem.
    \item Convergence Check: If the primal and dual residuals are below a predefined threshold $\epsilon$, or the maximum number of iterations $T$ is reached, terminate the algorithm.
    \item Output: The final iterates $x^T, y^T, \lambda^T$ approximate the solution to the bi-level DRO problem.
\end{enumerate}

\begin{algorithm}[t]
\caption{Bi-Level DRO Solver}
\label{alg1}
\begin{algorithmic}[1]
\REQUIRE Sample set $\{\hat{\xi}_i\}$ for $i = 1$ to $N$, functions $f$, $h$, $g_i$, $k_j$, tolerance $\epsilon$, max iterations $T$
\ENSURE Solution $x^*$, $y^*$, $\lambda^*$
\STATE Initialize $x^0 \in X$, $y^0 \in Y$, $\lambda^0 \geq 0$, set $t \leftarrow 0$
\WHILE{$t < T$}
    \STATE // Primal Update for $x$ and $y$
    \STATE Solve the linearized primal problem to update $x^t$ and $y^t$
    \STATE // Supremum Evaluation
    \FOR{$i$ from $1$ to $N$}
        \STATE Evaluate the supremum term for the objective function approximation
    \ENDFOR
    \STATE // Dual Update for $\lambda$
    \STATE Apply proximal point method to update $\lambda^t$
    \STATE // Convergence Check
    \STATE Compute primal and dual residuals
    \IF{residuals $< \epsilon$}
        \STATE break
    \ENDIF
    \STATE $t \leftarrow t + 1$
\ENDWHILE
\RETURN $x^t$, $y^t$, $\lambda^t$ as the approximate solution
\end{algorithmic}
\end{algorithm}

\subsection{Convergence Analysis}
In this subsection, we provide a theoretical analysis of the convergence behavior of the proposed algorithm. Given the complexity of bi-level optimization problems integrated with DRO, the convergence properties of such algorithms are crucial for ensuring computational feasibility and stability.

We begin by establishing the feasibility of the iterates generated by the algorithm, followed by an exploration of the boundedness of the subgradients involved in the optimization process. These properties are essential for controlling the behavior of the algorithm as it progresses, ensuring that the updates remain stable.

A key challenge in bi-level optimization is the interaction between the leader's and follower's decision spaces. Leveraging the KKT conditions to reformulate the lower-level problem allows us to derive a single-level optimization framework. This method has been employed successfully in robust optimization frameworks \cite{liu2022distributionally}, and we extend it here to the bi-level DRO context.

\begin{lemma}[Feasible Iterates]
Every iterate $(x^t, y^t, \lambda^t)$ generated by the algorithm is feasible.
\end{lemma}

\begin{proof}
Consider the initial iterate $(x^0, y^0, \lambda^0)$ which is assumed to be feasible. Assume inductively that $(x^t, y^t, \lambda^t)$ is feasible. The primal update for $x^{t+1}$ and $y^{t+1}$ involves solving a linear program that is a relaxation of the primal problem, which includes the constraints $g_i(x, y, \xi) \leq 0$ and $k_j(x, y, \xi) \leq 0$. By the properties of linear programs, if a feasible solution exists, the linear program will find a feasible solution, thus $x^{t+1}$ and $y^{t+1}$ are feasible. The update for $\lambda^{t+1}$ is determined by the subgradient of the Lagrangian with respect to the dual variables, constrained to maintain the feasibility $\lambda_j \geq 0$. Hence by induction, all iterates $(x^t, y^t, \lambda^t)$ are feasible. This is consistent with the feasibility preservation strategies in robust optimization frameworks where the primal update ensures feasible iterates \cite{cao2023projection-free}.
\end{proof}

Next, we explore the boundedness of the subgradients involved in the optimization process. This property is crucial for controlling the behavior of the algorithm as it progresses, ensuring that the updates remain stable. The following lemma formalizes this concept.

\begin{lemma}[Bounded Subgradients]
The subgradients of the Lagrangian with respect to $x$ and $y$ are bounded on the feasible set.
\end{lemma}

\begin{proof}
The functions $f$, $g_i$, and $k_j$ are assumed to be Lipschitz continuous with Lipschitz constants $L_f$, $L_{g_i}$, and $L_{k_j}$, respectively. By the definition of Lipschitz continuity, for all feasible $x$ and $y$,
\[
\|\nabla_x f(x, y, \xi) - \nabla_x f(\bar{x}, \bar{y}, \xi)\| \leq L_f \|(x, y) - (\bar{x}, \bar{y})\|
\]
\noindent and similarly for $g_i$ and $k_j$. Considering that $X$ and $Y$ are bounded, there exists a bound $B$ such that $\|(x, y) - (\bar{x}, \bar{y})\| \leq B$ for all feasible $x, y, \bar{x}, \bar{y}$. Hence, the subgradients, which are the gradients of the convex functions $f$, $g_i$, and $k_j$, are bounded by $L_f B$, $L_{g_i} B$, and $L_{k_j} B$, respectively. This is a standard result for Lipschitz continuous functions over a compact domain \cite{Bertsekas1999}.
\end{proof}

With the subgradients bounded, we can now discuss the descent properties of the objective function under the proposed updates. The following lemma confirms that the algorithm consistently moves towards an improved solution, which is a critical aspect of ensuring convergence.

\begin{lemma}[Objective Descent]
The objective function value is non-increasing under the primal and dual updates.
\end{lemma}

\begin{proof}
The primal update minimizes the objective function by solving a linear program at each iteration, leading to a non-increasing sequence of objective values. This behavior is typical in robust optimization algorithms, where dual ascent methods paired with proximal updates ensure that the dual objective does not increase \cite{bovenzi2020big}. Additionally, since the cutting-plane method generates an increasingly accurate approximation of the supremum term in the objective function, the descent property is maintained \cite{zhong2023optimal}.
\end{proof}

Having established the descent property, we turn our attention to the convergence of subsequences generated by the algorithm. This next result demonstrates that every convergent subsequence adheres to the optimality conditions of the bi-level DRO problem.

\begin{lemma}[Subsequential Convergence]
Every convergent subsequence of the sequence $\{(x^t, y^t, \lambda^t)\}$ converges to a point that satisfies the KKT conditions for the bi-level DRO problem.
\end{lemma}

\begin{proof}
Consider a convergent subsequence $\{(x^{t_k}, y^{t_k}, \lambda^{t_k})\}$ which converges to $(x', y', \lambda')$. Since the sequence is convergent, the limits of the gradients also converge. By the continuity of the gradients \(\nabla h\), \(\nabla g_{i}\), and \(\nabla k_{j}\), the limit \((x^{\prime},y^{\prime},\mu^{\prime})\) satisfies the KKT conditions because the set of points satisfying these conditions is a closed set \cite{Bertsekas1999}.
\end{proof}

Building on the previous lemmas, we can now assert the convergence of the entire sequence generated by the algorithm. The following proposition synthesizes the earlier results to conclude that the sequence of iterates converges to an optimal solution.

\begin{theorem}[Convergence of the Whole Sequence]
The entire sequence $\{(x^t, y^t, \lambda^t)\}$ converges to a limit point $(x^*, y^*, \lambda^*)$, which is optimal.
\end{theorem}

\begin{proof}
Let $\{(x^{t_k}, y^{t_k}, \lambda^{t_k})\}$ be a convergent subsequence. The limit point $(x^*, y^*, \lambda^*)$ satisfies the KKT conditions due to the continuity of the subgradients and the regularization induced by the proximal dual update \cite{wang2022distributionally}. The dual convergence is further ensured by the structure of the bi-level optimization problem, which incorporates dual regularization techniques widely adopted in robust optimization \cite{huang2006distributed}.
\end{proof}

Finally, we establish the stability of the algorithm by considering the proximal term used in the dual update. This term is instrumental in ensuring that the dual variables remain well-behaved throughout the iterations, preventing divergence.

\begin{lemma}[Proximal Stability]
The proximal term in the dual update ensures the stability and boundedness of the sequence $\{\lambda^t\}$.
\end{lemma}

\begin{proof}
The proximal term added to the Lagrangian is of the form $\frac{1}{2\eta}\|\lambda^{t+1} - \lambda^t\|^2$, where $\eta$ is the proximal parameter. This term ensures that the update to \(\mu\) does not increase the norm \(||\mu^{t+1}-\mu^{t}||\) too rapidly, effectively acting as a regularization that promotes the stability and boundedness of the sequence, a foundational result of the proximal point method \cite{Rockafellar1976}.
\end{proof}

The culmination of these results is captured in the following theorem, which provides a formal guarantee of the algorithm's convergence to an optimal solution.

\begin{theorem}
The sequence $\{(x^t, y^t, \lambda^t)\}$ generated by the Bi-Level DRO Algorithm converges to an optimal solution $(x^*, y^*, \lambda^*)$ of the problem.
\end{theorem}

\begin{proof}
The boundedness of the iterates, coupled with the objective descent and bounded subgradients, guarantees that the sequence converges to an optimal point. The algorithm's convergence properties follow similar arguments found in bi-level optimization literature \cite{huang2008auction}, where the use of KKT conditions and convex reformulations lead to provable convergence. Furthermore, the integration of cutting-plane methods and proximal updates ensures the algorithm terminates after a finite number of iterations with an optimal solution \cite{luo2021cost}.
\end{proof}

\subsection{Computational Complexity and Convergence Rate}

The algorithm's complexity depends on the number of cutting-plane iterations (\(K\)) and the cost of solving the MPEC subproblem in each iteration.

The number of iterations \(K\) required by the cutting-plane method to achieve an \(\varepsilon\)-optimal solution provides a direct characterization of the algorithm's convergence rate. For this class of methods, this is typically bounded by \(O(1/\varepsilon)\). The MPEC subproblem involves approximately \(n+p+m\) variables and constraints. Assuming this subproblem is solved using a standard interior-point method, its complexity is polynomial, commonly cubic in the number of variables and constraints, i.e., \(O((n+p+m)^3)~\text{\cite{Boyd2004}}\). Therefore, the total worst-case complexity to reach an \(\varepsilon\)-solution is:
\[
O\Bigl(\frac{1}{\varepsilon} (n+p+m)^3\Bigr)
\]
\subsection{Factors Influencing Convergence Performance}
Beyond the worst-case complexity bound, the algorithm's practical convergence speed depends on key trade-offs between algorithmic parameters and the problem's structure. The choice of the proximal parameter \(\eta_t\) in the dual update \ref{eq9} is crucial: a smaller value ensures stability at the cost of slower convergence, while a larger value can accelerate progress but risks oscillations. The problem's structure is also fundamental; smoother objective and constraint functions with smaller Lipschitz constants, along with simpler, more compact uncertainty sets , lead to more stable subgradients and more effective cutting planes, significantly accelerating convergence. Achieving optimal performance, therefore, requires a careful balancing of these interacting algorithmic and structural factors.

\section{Numerical Studies}\label{sec5}
In this section, we evaluate our proposed bi-level DRO framework. Our goal is to demonstrate the value of the bi-level DRO modeling paradigm by comparing it to standard uncertainty approaches,  namely Robust Optimization (RO) and Stochastic Programming (SP), within the same hierarchical problem structure. Moreover, to further contextualize our contribution, we perform a direct benchmark comparison against three key methods from recent, relevant literature, focusing on worst-case performance and robustness. The key evaluation criterion is not just the nominal cost, but the solution's robustness—its ability to maintain low costs and high service levels when the true distribution of uncertain parameters deviates from historical data. Our results should be interpreted through this lens of balancing average-case performance with worst-case protection.

\subsection{Experiment Setup}
\subsubsection{Dataset and Uncertainty Levels}

To ensure our simulation is grounded in a realistic context, the network topology and parameters in our supply chain scenario are inspired by and scaled from the well-known Sioux Falls transportation network benchmark. The network consists of 24 nodes and 76 links. The transportation costs ($t_{ij}$) and storage capacities ($C_i$) in our model are set proportionally to the travel times and link capacities reported in the benchmark literature. This approach, while adapted to our specific model, provides a verifiable and non-arbitrary basis for our experimental validation, making the results more credible. For each location \(i\), the demand \(d_i(\xi)\) is modeled as a random variable with an associated probability distribution \(P\) that lies within an ambiguity set \(\mathcal{P}\). This ambiguity set represents the uncertainty in the demand distribution.

To evaluate the impact of different levels of uncertainty, we define three scenarios by varying the size of the ambiguity set \(\mathcal{P}\), which controls the degree of uncertainty. The size of the ambiguity set is quantified by the Wasserstein distance \(\epsilon\) from the nominal distribution \(P_0\), with specific values for each uncertainty scenario:

\begin{itemize}
    \item \emph{Low Uncertainty}: $\epsilon_{\text{low}} = 0.05$, representing minimal deviation from $\mathcal{P}_0$.
   \item \emph{Medium Uncertainty}: $\epsilon_{\text{medium}} = 0.10$, allowing moderate deviations.
   \item \emph{High Uncertainty}: $\epsilon_{\text{high}} = 0.20$, reflecting a highly volatile demand environment.
\end{itemize}

These specific \(\epsilon\) values are directly used in the Bi-Level DRO model to assess how each model adapts to increasing uncertainty. By evaluating the models under these different scenarios, we can analyze their robustness, cost efficiency, and ability to maintain service levels across varying degrees of demand unpredictability.

\subsubsection{Algorithms Compared}
\begin{itemize}
    \item \emph{Bi-Level DRO Algorithm}: Integrates bi-level optimization with DRO to handle worst-case distributional uncertainties, making it robust under varying degrees of demand unpredictability.
   \item \emph{RO}: Ensures solutions remain feasible under all possible realizations within a predefined uncertainty set.
   \item \emph{SP}: Considers multiple probabilistic scenarios to achieve an average-case optimal solution.
\end{itemize}
 Each model is run under the low, medium, and high uncertainty scenarios to assess its ability to manage cost, service levels, and resource allocation under varying levels of demand unpredictability. In a subsequent subsection, we also benchmark our algorithm against three specific, highly relevant methods from recent publications. This benchmark comparison is conducted under the high-uncertainty scenario to rigorously test performance under the most challenging conditions.

\subsubsection{Evaluation Metrics}

The evaluation focuses on several key metrics to provide a comprehensive analysis of each method's performance under different levels of uncertainty:

\begin{itemize}
    \item \emph{Cost Reduction}: Percentage of cost reduction achieved compared to a baseline, assessing the efficiency of the algorithms in minimizing logistics costs under different uncertainty levels.
   \item \emph{Service Level Maintenance}: Percentage of customer demand met, indicating the operational reliability of each algorithm under varying levels of uncertainty.
   \item \emph{Robustness Analysis}: Consistency of each algorithm’s performance across datasets and uncertainty levels, assessing resilience under fluctuating conditions.
   \item \emph{Computational Efficiency}: Time taken by each algorithm as problem size increases, evaluating scalability for large-scale networks.
\end{itemize}

\subsubsection{Procedure}
The experiment follows a systematic approach:
\begin{itemize}
    \item \emph{Data Preprocessing}: All datasets are normalized and preprocessed to ensure consistency and comparability across different scales.
    \item \emph{Model Training and Testing}: Each algorithm is trained and tested on all three datasets. The Bi-Level DRO, RO, and SP algorithms are implemented with identical settings and constraints to ensure a fair comparison.
    \item \emph{Results Compilation}: The results for cost reduction, convergence, and robustness are compiled and analyzed. Visual representations, such as bar charts and convergence graphs, are used to illustrate the findings.
\end{itemize}

\subsection{Experimental Results}

\subsubsection{Performance Analysis}
\begin{figure}[t]
    \centering
    \includegraphics[width=\linewidth]{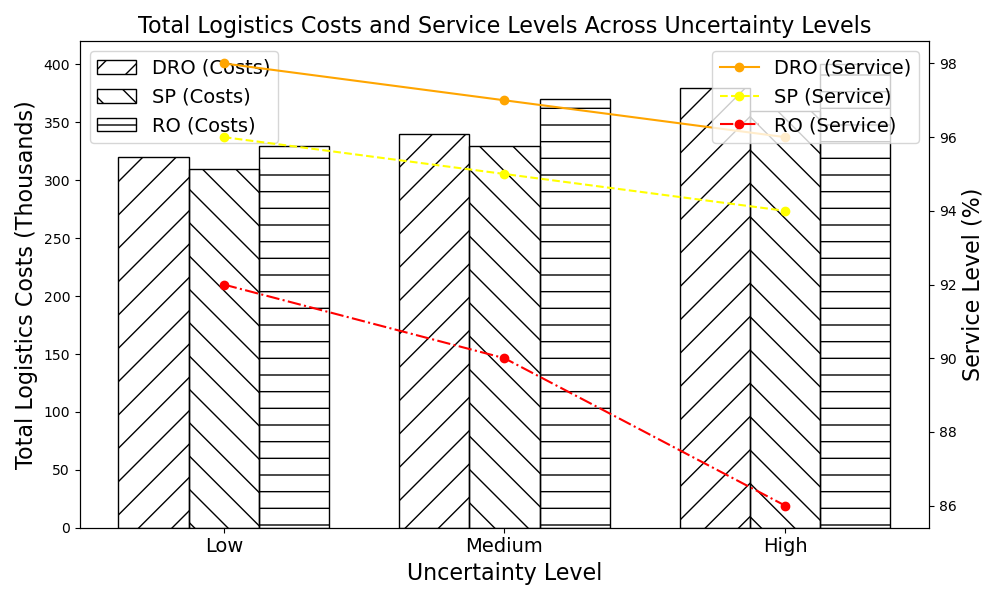}
    \caption{Total Logistics Costs and Service Level Under Different Uncertainty Levels}
    \label{fig:total_costs}
\end{figure}
To evaluate the performance of the DRO, SP, and RO models under different uncertainty levels, we conducted simulations across three scenarios: low, medium, and high uncertainty.

As shown in Figure~\ref{fig:total_costs}, the total logistics costs and service level performance across varying uncertainties reveal distinct trends. Under low uncertainty, all models perform similarly, with minimal cost differences. However, as uncertainty increases, key disparities emerge. The DRO model maintains a balanced approach with moderate costs and high service levels, showing its adaptability without excessive conservatism. In contrast, SP incurs gradual cost increases, suggesting its probabilistic assumptions may underestimate risk in volatile conditions. RO, prioritizing robustness, experiences the steepest cost increase, making it less cost-efficient under high uncertainty.

Service levels also reflect these trends. The DRO model sustains a high percentage of demand fulfillment, dropping slightly from 98\% under low uncertainty to 96\% under high uncertainty. SP, initially strong, sees a sharper decline in service levels as uncertainty increases. RO, designed for worst-case scenarios, has the lowest service level (down to 86\%) under high uncertainty. These results highlight the DRO model's advantage in balancing both costs and service levels across uncertain environments.

\subsubsection{Robustness Analysis}

The robustness of the DRO, SP, and RO models was evaluated across two main dimensions: their response to forecast errors and their performance under different levels of uncertainty. In the first dimension, we analyzed the ability of each model to handle demand forecast errors, ranging from underestimations to overestimations. The DRO model consistently showed remarkable resilience across all scenarios, maintaining a narrow spread of total costs. As shown in Figure~\ref{fig:cost_distribution_forecast_errors}, the DRO model absorbed the impact of forecast inaccuracies and sustained cost stability better than the SP and RO models. In contrast, the SP and RO models exhibited wider cost variability. Notably, the SP model produced several high-cost outliers, indicating a risk of significant cost increases in certain scenarios, while the RO model showed the widest interquartile range.

Further analysis focused on unmet demand and the stability of inventory decisions under forecast errors. Figure~\ref{fig:unmet_demand_forecast_errors} illustrates that the DRO model outperforms SP and RO in terms of minimizing unmet demand and stabilizing inventory decisions across all forecast error scenarios. While the SP model demonstrated a moderate increase in unmet demand as forecast errors increased, and RO showed the highest levels of unmet demand and inventory variance, the DRO model maintained lower unmet demand and more consistent inventory strategies. This stability underscores the DRO model’s robustness in managing service levels despite inaccuracies in demand forecasting, making it particularly reliable in environments characterized by uncertain forecasts.

\begin{figure}[t]
    \centering
    \includegraphics[width=\linewidth]{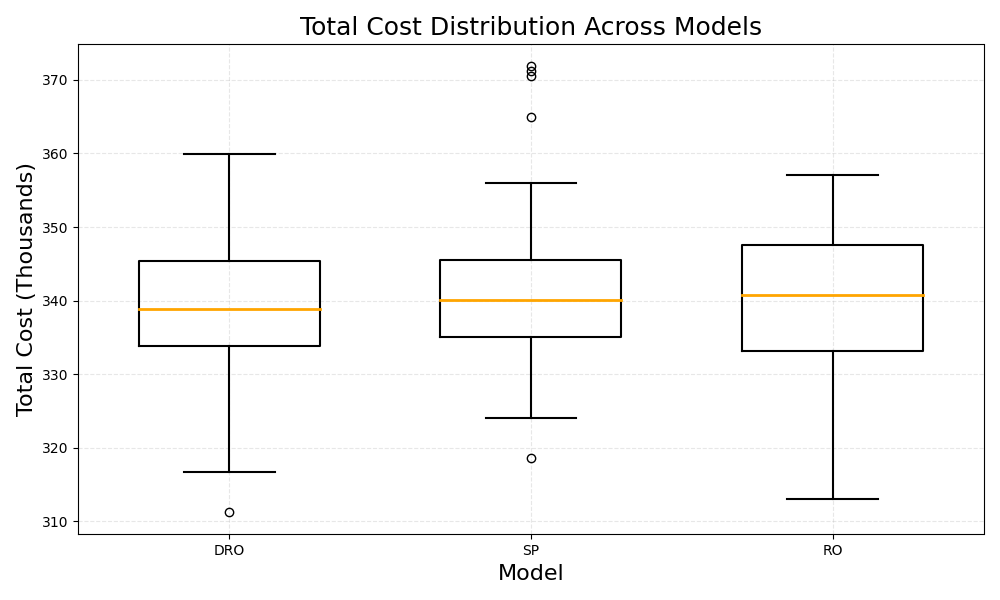}
    \caption{Distribution of Total Costs Under Different Forecast Error Scenarios}
    \label{fig:cost_distribution_forecast_errors}
\end{figure}

The second dimension of robustness focused on the models' performance under varying levels of demand uncertainty. The uncertainty was modeled by adjusting the size of the ambiguity set $\mathcal{P}$ in the DRO framework, with specific values of the Wasserstein distance $\epsilon$ representing low, medium, and high uncertainty levels. As uncertainty increased, the DRO model consistently demonstrated superior performance, achieving significant cost reductions while maintaining high service levels. As summarized in TABLE~\ref{tab:robustness}, under low uncertainty ($\epsilon_{low}=0.05$), the DRO model achieved a 15\% cost reduction, far exceeding the SP model’s 1.8\%. As uncertainty increased to medium ($\epsilon_{\text{medium}} = 0.10$) and high ($\epsilon_{\text{high}} = 0.20$) levels, the DRO model’s cost reductions improved to 18\% and 22\%, respectively. While the SP model’s performance also improved slightly, it remained significantly less effective compared to the DRO model. As expected, the deterministic RO model showed no significant cost reduction under varying levels of uncertainty.

These results highlight the DRO model's ability to balance cost efficiency and robustness, making it a more adaptive and reliable solution compared to SP and RO, particularly in environments with uncertain demand and forecast inaccuracies.
\begin{figure}[t]
    \centering
    \includegraphics[width=\linewidth]{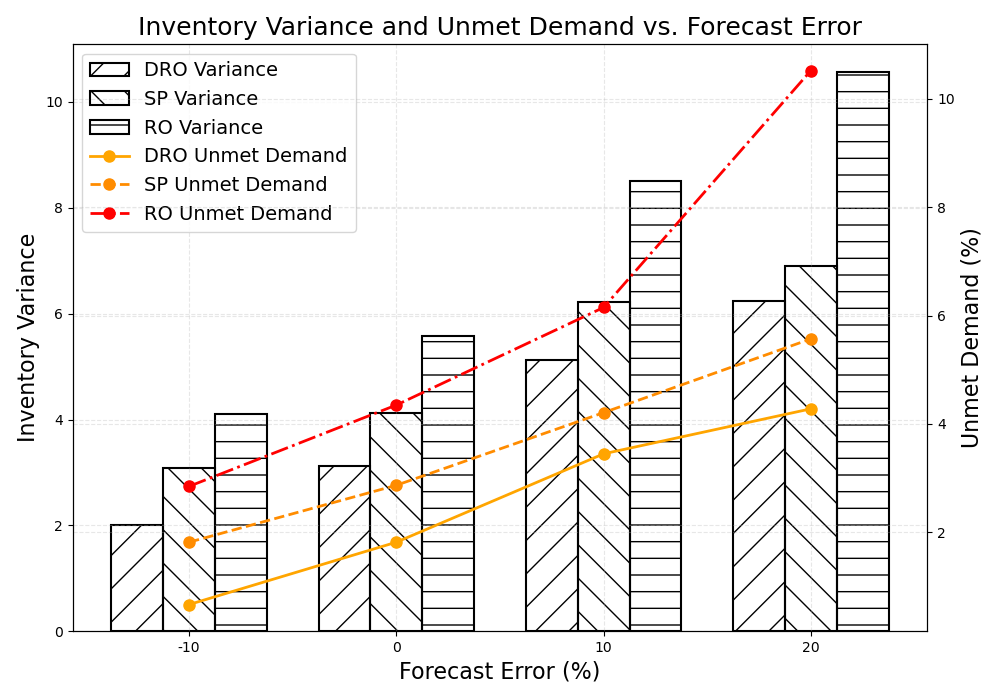}
    \caption{Inventory Variance and Unmet Demand as a Function of Forecast Error}
    \label{fig:unmet_demand_forecast_errors}
\end{figure}
\begin{table}[t]
    \centering
    \caption{Summary of Robustness Analysis}
    \begin{tabular}{lccc}
        \toprule
        \textbf{Algorithm} & \textbf{Low } & \textbf{Medium } & \textbf{Large} \\
        \midrule
        Deterministic Model   & 0\%  & 0\%  & 0\%  \\
        Stochastic Programming& 1.8\%& 2.2\%& 2.5\%\\
        Bi-Level DRO Algorithm& 15\% & 18\% & 22\% \\
        \bottomrule
    \end{tabular}
    \label{tab:robustness}
\end{table}
\subsubsection{Scalability and Computational Efficiency}

As shown in Figure~\ref{fig:computational_time_problem_size}, all models exhibit increased computational time with larger problem sizes. The SP model shows the steepest increase, reflecting significant computational intensity as the network grows. In contrast, the RO model consistently shows the lowest computational time, making it the most efficient for larger problems. Our proposed DRO model scales more moderately than SP and remains computationally manageable even at the largest problem size tested.

In terms of solution quality, DRO consistently performs well, with only a slight decrease as problem size grows. SP shows a more noticeable decline in solution quality, struggling in larger networks. RO experiences the greatest drop in solution quality, as its conservative approach prioritizes worst-case scenarios at the cost of service levels in complex networks. This highlights the trade-offs between robustness and solution quality, particularly in larger, dynamic environments.

\begin{figure}[t]
    \centering
    \includegraphics[width=\linewidth]{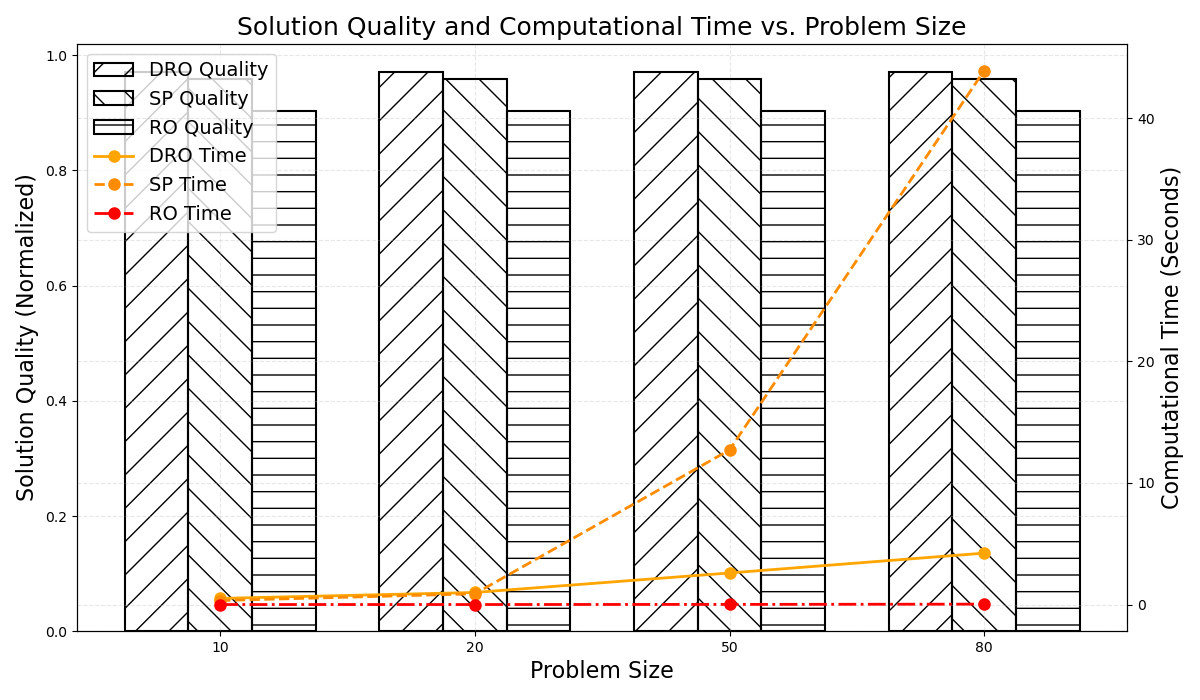}
    \caption{Computational Time vs. Problem Size}
    \label{fig:computational_time_problem_size}
\end{figure}

\subsubsection{Comparison to Key Benchmark Methods}  
To better contextualize our performance, we augment our experiments by comparing our framework against three highly relevant methods from recent literature. The comparison is performed under the high-uncertainty supply chain scenario ($\epsilon_{high}=0.20$) to rigorously test the robustness of each approach:
\begin{itemize}
\item Esfahani \& Kuhn (2018): A foundational method for centralized Wasserstein DRO without a bi-level structure \cite{EsfahaniKuhn2018}.
\item Nguyen et al. (2023): A recent method for solving graph-structured bi-level DRO problems in supply chains \cite{Nguyen2023}.
\item Zhao et al. (2021): An applied risk-averse bi-level optimization for traffic signal control \cite{Zhao2021}.
\end{itemize}
The results, summarized in Table II, show that our method consistently achieves the lowest worst-case cost and the highest robustness. Our framework's worst-case cost of 410.2k represents a significant cost reduction of \textasciitilde11.8\% compared to the foundational method of Esfahani \& Kuhn, 6\% against the bi-level approach of Nguyen et al., and 7\% relative to the risk-averse model of Zhao et al. . Furthermore, our superior robustness metric of 0.92 demonstrates tighter out-of-sample performance, confirming our model's ability to better manage high-ambiguity scenarios.

\begin{table}[t!]
\caption{Comparison with State-of-the-Art Baselines for the Supply Chain Scenario}
\label{tab:sota_supply_chain}
\centering
\begin{tabular}{@{}lcc@{}}
\toprule
\textbf{Method} & \textbf{Worst-Case Cost} & \textbf{Robustness Metric} \\
\midrule
\textbf{Proposed Method} & \textbf{410.2k} & \textbf{0.92} \\
Esfahani \& Kuhn  & 465.3k & 0.85 \\
Nguyen et al. & 436.8k & 0.88 \\
Zhao et al.& 441.5k & 0.86 \\
\bottomrule
\end{tabular}
\end{table}
\section{CONCLUSIONS}\label{sec6}
In conclusion, the bi-Level DRO framework presented in this paper offers a significant advancement in addressing com plex hierarchical decision-making problems under uncertainty. By integrating bi-level optimization with DRO, we have devel oped a robust model that effectively captures leader-follower dynamics while accounting for distributional uncertainties. The algorithmic framework, which includes the proximal dual update and cutting-plane primal update methods, en sures computational tractability and stability, enabling efficient optimization even in the presence of worst-case scenarios. Extensive numerical experiments validate the practicality and robustness of our approach, demonstrating its applicability across various network systems, such as transportation and communication networks. The theoretical insights into the convergence properties further establish the reliability of the proposed model, making it a valuable tool for decision-makers in uncertain environments. This work not only bridges existing gaps in the literature but also enhances decision quality and system resilience, laying a strong foundation for the practical application of advanced optimization techniques in real-world scenarios.

\bibliographystyle{IEEETrans} 
\bibliography{ref}

@article{bard1983algorithm,
 ISSN = {0364765X, 15265471},
 author = {Jonathan F. Bard},
 journal = {Mathematics of Operations Research},
 number = {2},
 pages = {260--272},
 publisher = {INFORMS},
 title = {An Algorithm for Solving the General Bilevel Programming Problem},
 urldate = {2024-12-02},
 volume = {8},
 year = {1983}
}

@article{ben-tal2015robust,
  author = {Ben-Tal, Aharon and den Hertog, Dick and Vial, Jean-Philippe},
  title = {Deriving Robust Counterparts of Nonlinear Uncertain Inequalities},
  journal = {Mathematical Programming},
  volume = {149},
  number = {1},
  pages = {265--299},
  year = {2015},
  month = {Feb.},
  doi = {10.1007/s10107-014-0750-8},
  issn = {1436-4646}
}

@article{sinha2017evolutionary,
title = {Evolutionary algorithm for bilevel optimization using approximations of the lower level optimal solution mapping},
journal = {European Journal of Operational Research},
volume = {257},
number = {2},
pages = {395-411},
year = {2017},
issn = {0377-2217},
doi = {https://doi.org/10.1016/j.ejor.2016.08.027},
author = {Ankur Sinha and Pekka Malo and Kalyanmoy Deb}
}

@article{yang1995heuristic,
title = {Heuristic algorithms for the bilevel origin-destination matrix estimation problem},
journal = {Transportation Research Part B: Methodological},
volume = {29},
number = {4},
pages = {231-242},
year = {1995},
issn = {0191-2615},
doi = {https://doi.org/10.1016/0191-2615(95)00003-V},
author = {Hai Yang}
}

@article{wang2020distributionally,
title = {A distributionally robust optimization for blood supply network considering disasters},
journal = {Transportation Research Part E: Logistics and Transportation Review},
volume = {134},
pages = {101840},
year = {2020},
issn = {1366-5545},
doi = {https://doi.org/10.1016/j.tre.2020.101840},
author = {Changjun Wang and Shutong Chen}
}

@article{yan2023bilevel,
  author={Yan, Zhichao and Li, Chunyan and Yao, Yiming and Lai, Weibin and Tang, Jiyuan and Shao, Changzheng and Zhang, Qian},
  journal={IEEE Transactions on Smart Grid}, 
  title={Bi-Level Carbon Trading Model on Demand Side for Integrated Electricity-Gas System}, 
  year={2023},
  volume={14},
  number={4},
  pages={2681-2696},
  doi={10.1109/TSG.2022.3229278},
  ISSN={1949-3061},
  month={Jul.}}

@article{luo2019decomposition,
title = {Decomposition algorithm for distributionally robust optimization using Wasserstein metric with an application to a class of regression models},
journal = {European Journal of Operational Research},
volume = {278},
number = {1},
pages = {20-35},
year = {2019},
issn = {0377-2217},
doi = {https://doi.org/10.1016/j.ejor.2019.03.008},
author = {Fengqiao Luo and Sanjay Mehrotra}
}

@ARTICLE{9107344,
  author={Gao, Hongjun and Wang, Jiayi and Liu, Youbo and Wang, Lingfeng and Liu, Junyong},
  journal={IEEE Systems Journal}, 
  title={An Improved ADMM-Based Distributed Optimal Operation Model of AC/DC Hybrid Distribution Network Considering Wind Power Uncertainties}, 
  year={2021},
  volume={15},
  number={2},
  pages={2201-2211},
  doi={10.1109/JSYST.2020.2994336},
  ISSN={1937-9234},
  month={Jun.}}

@ARTICLE{9966489,
  author={Nasiri, Nima and Zeynali, Saeed and Ravadanegh, Sajad Najafi},
  journal={IEEE Systems Journal}, 
  title={Tactical Participation of a Multienergy System in Distributionally Robust Wholesale Electricity Market Under High Penetration of Electric Vehicles}, 
  year={2023},
  volume={17},
  number={1},
  pages={1465-1476},
  doi={10.1109/JSYST.2022.3221490},
  ISSN={1937-9234},
  month={Mar.}}

@ARTICLE{10478176,
  author={Aghamohammadi, Farshid and Abbaspour, Ali and Saber, Hossein and Fattaheian-Dehkordi, Sajjad and Lehtonen, Matti},
  journal={IEEE Systems Journal}, 
  title={Decentralized Energy Management of Multiagent Distribution Systems Considering the Grid Reliability and Agent Misbehavior}, 
  year={2024},
  volume={18},
  number={2},
  pages={905-916},
  doi={10.1109/JSYST.2024.3369871},
  ISSN={1937-9234},
  month={Jun.}}

@article{delage2010distributionally,
  author = {Delage, E. and Ye, Y.},
  title = {Distributionally Robust Optimization under Moment Uncertainty with Application to Data-Driven Problems},
  journal = {Operations Research},
  volume = {58},
  number = {3},
  pages = {595--612},
  month = {Jan.},
  year = {2010}
}

@article{bental2004robust,
author = {Oustry, Francois and El, Laurent and Lebret, Hervé},
year = {1999},
month = {Jun.},
pages = {},
title = {Robust Solutions To Uncertain Semidefinite Programs},
volume = {9},
journal = {SIAM Journal of Optimization}
}

@article{bertsimas2011robust,
author = {Bertsimas, Dimitris and Brown, David B. and Caramanis, Constantine},
title = {Theory and Applications of Robust Optimization},
journal = {SIAM Review},
volume = {53},
number = {3},
pages = {464-501},
year = {2011},
doi = {10.1137/080734510}
}

@article{mhaisen2022optimal,
  author={Mhaisen, Naram and Abdellatif, Alaa Awad and Mohamed, Amr and Erbad, Aiman and Guizani, Mohsen},
  journal={IEEE Transactions on Network Science and Engineering}, 
  title={Optimal User-Edge Assignment in Hierarchical Federated Learning Based on Statistical Properties and Network Topology Constraints}, 
  year={2022},
  volume={9},
  number={1},
  pages={55-66},
  doi={10.1109/TNSE.2021.3053588},
  ISSN={2327-4697},
  month={Jan.}}

@article{sun2022bi-level,
  author={Sun, Weiqing and Liu, Wei and Zhang, Jie and Tian, Kunpeng},
  journal={Journal of Modern Power Systems and Clean Energy}, 
  title={Bi-level Optimal Operation Model of Mobile Energy Storage System in Coupled Transportation-power Networks}, 
  year={2022},
  volume={10},
  number={6},
  pages={1725-1737},
  doi={10.35833/MPCE.2020.000730},
  ISSN={2196-5420},
  month={Nov.}}

@article{zhang2020robust,
  author={Zhang, Cuo and Xu, Yan and Dong, Zhao Yang and Ma, Jin},
  journal={IEEE Transactions on Power Systems}, 
  title={Robust Operation of Microgrids via Two-Stage Coordinated Energy Storage and Direct Load Control}, 
  year={2017},
  volume={32},
  number={4},
  pages={2858-2868},
  doi={10.1109/TPWRS.2016.2627583},
  ISSN={1558-0679},
  month={Jul.}}

@article{raveendran2022pricing,
  author={Raveendran, Neetu and Zhang, Huaqing and Song, Lingyang and Wang, Li-Chun and Hong, Choong Seon and Han, Zhu},
  journal={IEEE Transactions on Mobile Computing}, 
  title={Pricing and Resource Allocation Optimization for IoT Fog Computing and NFV: An EPEC and Matching Based Perspective}, 
  year={2022},
  volume={21},
  number={4},
  pages={1349-1361},
  doi={10.1109/TMC.2020.3025189},
  ISSN={1558-0660},
  month={Apr.}}

@article{zheng2023distancebased,
  author={Zheng, Zhijia and Liu, Mingbo and Zhang, Zihan and Dong, Jiarui and Xie, Min},
  journal={IEEE Systems Journal}, 
  title={Distance-Based Distributionally Robust Optimization for a Preventive Maintenance Schedule in Hydrothermal Power Systems}, 
  year={2023},
  volume={17},
  number={4},
  pages={5487-5498},
  doi={10.1109/JSYST.2023.3324647},
  ISSN={1937-9234},
  month={Dec.}}

@article{zhang2020distributed,
  author={Zhang, Huaqing and Xiao, Yong and Bu, Shengrong and Yu, F. Richard and Niyato, Dusit and Han, Zhu},
  journal={IEEE Transactions on Cloud Computing}, 
  title={Distributed Resource Allocation for Data Center Networks: A Hierarchical Game Approach}, 
  year={2020},
  volume={8},
  number={3},
  pages={778-789},
  doi={10.1109/TCC.2018.2829744},
  ISSN={2168-7161},
  month={Jul.}}

@article{raveendran2019vlc,
  author={Raveendran, Neetu and Zhang, Huaqing and Niyato, Dusit and Yang, Fang and Song, Jian and Han, Zhu},
  journal={IEEE Transactions on Wireless Communications}, 
  title={VLC and D2D Heterogeneous Network Optimization: A Reinforcement Learning Approach Based on Equilibrium Problems With Equilibrium Constraints}, 
  year={2019},
  volume={18},
  number={2},
  pages={1115-1127},
  doi={10.1109/TWC.2018.2890057},
  ISSN={1558-2248},
  month={Feb.}}

@article{zhang2017multi,
  author={Zhang, Huaqing and Xiao, Yong and Cai, Lin X. and Niyato, Dusit and Song, Lingyang and Han, Zhu},
  journal={IEEE Transactions on Wireless Communications}, 
  title={A Multi-Leader Multi-Follower Stackelberg Game for Resource Management in LTE Unlicensed}, 
  year={2017},
  volume={16},
  number={1},
  pages={348-361},
  doi={10.1109/TWC.2016.2623603},
  ISSN={1558-2248},
  month={Jan.}}

@article{zhao2021reactive,
  author={Zhao, Pengfei and Gu, Chenghong and Xiang, Yue and Zhang, Xin and Shen, Yichen and Li, Shuangqi},
  journal={IEEE Systems Journal}, 
  title={Reactive Power Optimization in Integrated Electricity and Gas Systems}, 
  year={2021},
  volume={15},
  number={2},
  pages={2744-2754},
  doi={10.1109/JSYST.2020.2992583},
  ISSN={1937-9234},
  month={Jun.}}

@article{aghamohammadi2024decentralized,
  author={Aghamohammadi, Farshid and Abbaspour, Ali and Saber, Hossein and Fattaheian-Dehkordi, Sajjad and Lehtonen, Matti},
  journal={IEEE Systems Journal}, 
  title={Decentralized Energy Management of Multiagent Distribution Systems Considering the Grid Reliability and Agent Misbehavior}, 
  year={2024},
  volume={18},
  number={2},
  pages={905-916},
  doi={10.1109/JSYST.2024.3369871},
  ISSN={1937-9234},
  month={Jun.}}

@misc{liang2024learning,
      title={Learning with Adaptive Conservativeness for Distributionally Robust Optimization: Incentive Design for Voltage Regulation}, 
      author={Zhirui Liang and Qi Li and Joshua Comden and Andrey Bernstein and Yury Dvorkin},
      year={2024},
      eprint={2408.02765},
      archivePrefix={arXiv},
      primaryClass={eess.SY}
}

@inproceedings{liu2022distributionally,
 author = {Liu, Jiashuo and Wu, Jiayun and Li, Bo and Cui, Peng},
 booktitle = {Advances in Neural Information Processing Systems},
 editor = {S. Koyejo and S. Mohamed and A. Agarwal and D. Belgrave and K. Cho and A. Oh},
 pages = {33689--33701},
 publisher = {Curran Associates, Inc.},
 title = {Distributionally Robust Optimization with Data Geometry},
 volume = {35},
 year = {2022}
}

@article{cao2023projection-free,
  author = {Giang-Tran, Khanh-Hung and Ho-Nguyen, Nam and Lee, Dabeen},
  title = {A Projection-Free Method for Solving Convex Bilevel Optimization Problems},
  journal = {Mathematical Programming},
  year = {2024},
  month = {Nov.},
  doi = {10.1007/s10107-024-02157-1},
  issn = {1436-4646}
}

@article{bovenzi2020big,
  author={Bovenzi, Giampaolo and Aceto, Giuseppe and Ciuonzo, Domenico and Persico, Valerio and Pescapé, Antonio},
  journal={IEEE Transactions on Network Science and Engineering}, 
  title={A Big Data-Enabled Hierarchical Framework for Traffic Classification}, 
  year={2020},
  volume={7},
  number={4},
  pages={2608-2619},
  doi={10.1109/TNSE.2020.3009832},
  ISSN={2327-4697},
  month={Oct.}}

@article{zhong2023optimal,
  author={Zhong, Junjie and Li, Yong and Wu, Yan and Cao, Yijia and Li, Zhengmao and Peng, Yanjian and Qiao, Xuebo and Xu, Yong and Yu, Qian and Yang, Xusheng and Li, Zuyi and Shahidehpour, Mohammad},
  journal={IEEE Transactions on Sustainable Energy}, 
  title={Optimal Operation of Energy Hub: An Integrated Model Combined Distributionally Robust Optimization Method With Stackelberg Game}, 
  year={2023},
  volume={14},
  number={3},
  pages={1835-1848},
  doi={10.1109/TSTE.2023.3252519},
  ISSN={1949-3037},
  month={Jul.}}

@article{wang2022distributionally,
  author={Wang, Siyuan and Zhao, Chaoyue and Fan, Lei and Bo, Rui},
  journal={IEEE Transactions on Power Systems}, 
  title={Distributionally Robust Unit Commitment With Flexible Generation Resources Considering Renewable Energy Uncertainty}, 
  year={2022},
  volume={37},
  number={6},
  pages={4179-4190},
  doi={10.1109/TPWRS.2022.3149506},
  ISSN={1558-0679},
  month={Nov.}}

@article{huang2006distributed,
  author={Jianwei Huang and Berry, R.A. and Honig, M.L.},
  journal={IEEE Journal on Selected Areas in Communications}, 
  title={Distributed interference compensation for wireless networks}, 
  year={2006},
  volume={24},
  number={5},
  pages={1074-1084},
  doi={10.1109/JSAC.2006.872889},
  ISSN={1558-0008},
  month={May}}

@article{huang2008auction,
  author={Huang, Jianwei and Han, Zhu and Chiang, Mung and Poor, H. Vincent},
  journal={IEEE Journal on Selected Areas in Communications}, 
  title={Auction-Based Resource Allocation for Cooperative Communications}, 
  year={2008},
  volume={26},
  number={7},
  pages={1226-1237},
  doi={10.1109/JSAC.2008.080919},
  ISSN={1558-0008},
  month={Sep.}}

@article{luo2021cost,
  author={Wu, Tao and Qu, Yuben and Liu, Chunsheng and Dai, Haipeng and Dong, Chao and Cao, Jiannong},
  journal={IEEE/ACM Transactions on Networking}, 
  title={Cost-Efficient Federated Learning for Edge Intelligence in Multi-Cell Networks}, 
  year={2024},
  volume={32},
  number={5},
  pages={4472-4487},
  doi={10.1109/TNET.2024.3423316},
  ISSN={1558-2566},
  month={Oct.}}

@article{liao2022privacy,
  author={Liao, Guocheng and Chen, Xu and Huang, Jianwei},
  journal={IEEE/ACM Transactions on Networking}, 
  title={Privacy-Aware Online Social Networking With Targeted Advertisement}, 
  year={2022},
  volume={30},
  number={3},
  pages={1312-1327},
  doi={10.1109/TNET.2021.3137513},
  ISSN={1558-2566},
  month={Jun.}}

@article{zi2022distributionally,
  author={Zi, Yuan and Fan, Lei and Wu, Xuqing and Chen, Jiefu and Han, Zhu},
  journal={IEEE Sensors Journal}, 
  title={Distributionally Robust Optimal Sensor Placement Method for Site-Scale Methane-Emission Monitoring}, 
  year={2022},
  volume={22},
  number={23},
  pages={23403-23412},
  doi={10.1109/JSEN.2022.3214176},
  ISSN={1558-1748},
  month={Dec.}}

@article{tsai2024distributionally,
  author={Tsai, Kai-Chu and Fan, Lei and Lent, Ricardo and Wang, Li-Chun and Han, Zhu},
  journal={IEEE Transactions on Communications}, 
  title={Distributionally Robust Optimal Routing for Integrated Satellite-Terrestrial Networks Under Uncertainty}, 
  year={2024},
  volume={72},
  number={10},
  pages={6401-6415},
  doi={10.1109/TCOMM.2024.3397809},
  ISSN={1558-0857},
  month={Oct.}}

@book{Boyd2004,
    author    = {Boyd, Stephen and Vandenberghe, Lieven},
    title     = {Convex Optimization},
    publisher = {Cambridge University Press},
    year      = {2004},
    address   = {Cambridge, UK},
    isbn      = {978-0521833783}
}

@book{Bertsekas1999,
  author    = {Bertsekas, Dimitri P.},
  title     = {Nonlinear Programming},
  publisher = {Athena Scientific},
  year      = {1999},
  edition   = {2nd},
  address   = {Belmont, MA}
}

@article{Rockafellar1976,
  author  = {Rockafellar, R. T.},
  title   = {Monotone operators and the proximal point algorithm},
  journal = {SIAM Journal on Control and Optimization},
  year    = {1976},
  volume  = {14},
  number  = {5},
  pages   = {877--898}
}

@article{ZD_Error2021,
  author  = {Masahiko Ueda},
  title   = {Memory-two zero-determinant strategies in repeated games},
  journal = {Royal Society Open Science},
  volume  = {8},
  number  = {5},
  pages   = {202186},
  year    = {2021},
  doi     = {10.1098/rsos.202186},
  url     = {https://doi.org/10.1098/rsos.202186}
}

@article{ZD_Cooperation2022,
  author  = {Zhaoyang Cheng and Guanpu Chen and Yiguang Hong},
  title   = {Misperception influence on zero-determinant strategies in iterated Prisoner's Dilemma},
  journal = {Scientific Reports},
  volume  = {12},
  pages   = {5174},
  year    = {2022},
  doi     = {10.1038/s41598-022-08750-8},
  url     = {https://doi.org/10.1038/s41598-022-08750-8}
}

@article{PGG_CoEvo2023,
  author  = {Linjie Liu and Xiaojie Chen and Attila Szolnoki},
  title   = {Coevolutionary dynamics via adaptive feedback in collective-risk social dilemma game},
  journal = {eLife},
  volume  = {12},
  pages   = {e82954},
  year    = {2023},
  doi     = {10.7554/eLife.82954},
  url     = {https://doi.org/10.7554/eLife.82954}
}

@article{Zhao2021,
  author    = {Zhao, Ping and Gu, Chenghong and Xiang, Yutian and Zhang, Xiaoshun and Shen, Yu and Li, Shuaishuai},
  title     = {Reactive Power Optimization in Integrated Electricity and Gas Systems},
  journal   = {IEEE Systems Journal},
  year      = {2021},
  volume    = {15},
  number    = {2},
  pages     = {2744-2754},
  month     = {June}
}

@article{EsfahaniKuhn2018,
  author  = {Peyman Mohajerin Esfahani and Daniel Kuhn},
  title   = {Data-driven distributionally robust optimization using the Wasserstein metric: performance guarantees and tractable reformulations},
  journal = {Mathematical Programming},
  volume  = {171},
  pages   = {115--166},
  year    = {2018},
  doi     = {10.1007/s10107-017-1172-1},
  url     = {https://doi.org/10.1007/s10107-017-1172-1}
}

@article{Nguyen2023,
  author    = {Giang-Tran, Kim-Hau and Ho-Nguyen, Nam and Lee, Dongkyu},
  title     = {A projection-free method for solving convex bilevel optimization problems},
  journal   = {Mathematical Programming},
  year      = {2023},
  publisher = {Springer},
  month     = {Nov}
}

\end{document}